\documentclass[journal]{IEEEtran}

\makeatletter
\def\normalsize{\@setfontsize{\normalsize}{9.5pt}{11.5pt}}
\normalsize
\makeatother

\usepackage[utf8]{inputenc}
\usepackage{times}
\usepackage{algorithm}
\usepackage{algpseudocode}
\usepackage{graphicx}
\usepackage{amsmath}
\usepackage{amsfonts}
\usepackage{amsthm}
\usepackage{amssymb}
\usepackage{nicefrac}
\usepackage{psfrag}
\usepackage{multirow}
\usepackage{wasysym}
\usepackage{enumerate}
\usepackage[inline]{enumitem}
\usepackage[dvipsnames, usenames]{xcolor}
\usepackage{booktabs}[]
\usepackage{multirow}
\usepackage{soul}
\usepackage{etoolbox}
\usepackage{bm}
\usepackage{optidef}
\usepackage{subfig}
\usepackage{placeins}

\newtheorem{lemma}[]{{Lemma}}

\renewcommand{\vec}[1]{\ensuremath{\boldsymbol{{#1}}}}
\newcommand{\mat}[1]{\ensuremath{\boldsymbol{{#1}}}}

\newcommand{\RR}{\mathbb{R}}

\newcommand{\Exp}{{\mathbf{E}}}
\newcommand{\dd}{{\rm d}}

\newcommand{\tr}{{\rm tr}}
\newcommand{\diag}{{\rm diag}}

\newcommand{\x}{{\vec x}}
\newcommand{\y}{{\vec y}}
\newcommand{\z}{{\vec z}}
\newcommand{\q}{{\vec q}}

\newcommand{\Deltav}{{\vec \Delta}} 
\newcommand{\xiv}{{\vec \xi}}
\newcommand{\m}{{\vec m}}
\newcommand{\zb}{\bar{z}}
\newcommand{\zvb}{\bar{\z}}

\newcommand{\ok}{{\rm ok}}
\newcommand{\ko}{{\rm ko}}
\newcommand{\xok}{\x^\ok}
\newcommand{\xko}{\x^\ko}
\newcommand{\xh}{\hat{\x}}

\newcommand{\yok}{\y^{\ok}}

\newcommand{\zok}{\z^{\ok}}

\newcommand{\xhok}{\xh^\ok}
\newcommand{\xhko}{\xh^\ko}

\newcommand{\cov}{{\mat \Sigma}}
\newcommand{\covok}{\cov^\ok}
\newcommand{\covko}{\cov^\ko}
\newcommand{\lambdaok}{\lambda^\ok}
\newcommand{\lambdako}{\lambda^\ko}
\newcommand{\lambdav}{{\vec \lambda}}
\newcommand{\lambdavok}{\lambdav^\ok}
\newcommand{\lambdavko}{\lambdav^\ko}

\newcommand{\fok}{f^\ok}
\newcommand{\fko}{f^\ko}
\newcommand{\white}{{\mat W}}
\newcommand{\muok}{\mu^\ok}
\newcommand{\muv}{{\vec \mu}}
\newcommand{\muvok}{\muv^\ok}

\newcommand{\fx}{f_{\x}}
\newcommand{\fxh}{f_{\xh}}
\newcommand{\fokx}{f^\ok_{\x}}
\newcommand{\fkox}{f^\ko_{\x}}
\newcommand{\fokxh}{f^\ok_{\xh}}
\newcommand{\fkoxh}{f^\ko_{\xh}}
\newcommand{\fxhgx}{f_{\xh|{\x}}}

\newcommand{\mok}{\vec \mu^{\ok}}
\newcommand{\mko}{\vec \mu^{\ko} }

\newcommand{\covokxh}{\hat{\cov}^\ok}
\newcommand{\covkoxh}{\hat{\cov}^\ko}

\newcommand{\stack}{r}

\newcommand{\Pdet}{P_{\rm det}}

\newcommand{\Gauss}[2]{{\mathcal{G}}\left(#1,#2\right)}

\newcommand{\dist}{\omega}

\newcommand{\AUC}{{\rm AUC}}

\definecolor{dark}{HTML}{404070}
\definecolor{light}{HTML}{E0E0FF}

\title{RDD: Pareto Analysis of the Rate-Distortion-Distinguishability Trade-off}

\author{
    Andriy Enttsel,~\IEEEmembership{Student Member,~IEEE,}
    Alex~Marchioni,~\IEEEmembership{Member,~IEEE,}
    Andrea Zanellini,
    Mauro~Mangia,~\IEEEmembership{Member,~IEEE,}
    Gianluca~Setti,~\IEEEmembership{Fellow,~IEEE,}
    and~Riccardo~Rovatti,~\IEEEmembership{Fellow,~IEEE}
\thanks{A. Enttsel, A. Marchioni, M. Mangia, and R. Rovatti are with the Department of Electrical, Electronic, and Information Engineering, University of Bologna, 40136 Bologna, Italy, and also with the Advanced Research Center on Electronic Systems, University of Bologna, 40125 Bologna, Italy (e-mail: alex.marchioni@unibo.it, andriy.enttsel@unibo.it, andrea.zanellini2@unibo.it, mauro.mangia@unibo.it, riccardo.rovatti@unibo.it).}%
\thanks{G. Setti is with CEMSE, King Abdullah University of Science and Technology (KAUST), Saudi Arabia (e-mail: gianluca.setti@kaust.edu.sa). He is also on leave from the Department of Electronics and Telecommunications (DET), Politecnico di Torino, 10129 Torino.}%
\thanks{R. Rovatti is also with the Alma Mater Research Institute for Human-Centered AI University of Bologna, 40015 Bologna, Italy}
}

\begin{document}

\maketitle

\begin{abstract}
Extensive monitoring systems generate data that is usually compressed for network transmission. This compressed data might then be processed in the cloud for tasks such as anomaly detection. However, compression can potentially impair the detector's ability to distinguish between regular and irregular patterns due to information loss.
Here we extend the information-theoretic framework introduced in \cite{Marchioni_2024TSMC} to simultaneously address the trade-off between the three features on which the effectiveness of the system depends: the effectiveness of compression, the amount of distortion it introduces, and the distinguishability between compressed normal signals and compressed anomalous signals.
We leverage a Gaussian assumption to draw curves showing how moving on a Pareto surface helps administer such a trade-off better than simply relying on optimal rate-distortion compression and hoping that compressed signals can be distinguished from each other.

\end{abstract}


\section{Introduction}

The ever-growing pervasiveness of sensor networks and Internet-of-Things (IoT) devices necessitates the development of efficient acquisition systems. These systems play a critical role in various domains, from monitoring industrial processes and critical infrastructure to scientific research and environmental surveillance. By collecting and storing vast amounts of data, acquisition systems provide a comprehensive picture of the monitored environment, enabling real-time analysis and informed decision-making. However, the sheer volume of data generated by such systems poses challenges in storage, transmission, and computational processing.

In modern massive acquisition systems, we can identify a common scenario that involves numerous sensing units that convert unknown physical quantities into samples that are subsequently transmitted through a network. To minimize the amount of transmitted data and retain only essential information, sensor readings are often subjected to compression through a lossy mechanism \cite{
Loia_2017TSMC,
Liu_2022TSMC}. This process inherently involves a trade-off between the transmission bit \textit{rate} and the extent of information loss measured with \textit{distortion}.

The compressed data are then often transmitted to a remote unit -- potentially a cloud facility -- where it is processed and stored according to the application's needs. One such task is the detection of anomalies, which consists in determining whether the signal comes from a normal or an anomalous source \cite{
Burrello_JIoT2020,
Kurt_TPAMI2021,
Ciancarelli_IAC2022,
Lo_TPAMI2023,
Ciancarelli_SCI2023}. When the statistical characteristics of the anomalies are not known, the detection is an unsupervised problem, conversely, it can be viewed as a supervised task. 

Nevertheless, compression may affect the performance of a detector. In fact, lossy compression is most effective when it discards irrelevant information to the recovery of the typical signal, i.e., when the distortion is minimized. However, that information may be crucial for a detector to determine whether the signal is anomalous. As a result, designing a system that involves both compression and detection is a challenge that extends beyond merely striking a balance between rate and distortion, but it involves a threefold trade-off: \textit{rate-distortion-distinguishability} (RDD).

How a compression scheme addressing the rate-distortion trade-off affects distinguishability has already been investigated in \cite{Marchioni_2024TSMC}. Here, we extend that framework to consider the threefold trade-off. We therefore study the trade-off between compression, information loss, and the ability to discriminate between compressed normal and anomalous sources, in both supervised and unsupervised scenarios. To derive theoretical results, we rely on the Gaussian assumption for which both normal and anomalous sources and their compressed versions are considered Gaussian. The outcome consists of a Pareto surface that generalizes the rate-distortion curve \cite{Cover_1991} and optimally describes the rate-distortion-distinguishability trade-off in the Gaussian case.

We then replicate the trends highlighted by the theoretical results for real-world compression schemes and real-world signals. In particular, we consider:
\begin{enumerate*}[label=(\roman*), font=\itshape]
  \item a Karhunen-Loève transform (KLT)-based compression scheme working on Gaussian signals;
  \item JPEG compression in which quantization tables are adapted to promote anomaly detection;
  \item a Neural Network-based lossy compression integrating a term accounting for anomaly detection performance.
\end{enumerate*}
In all these cases, we highlight the dependence of anomaly detection performance on both rate and distortion.

The paper is organized as follows. In Section~\ref{sec:RDD} we extend the rate-distortion framework, including distinguishability measures in the optimization process. Section~\ref{sec:SOA} is dedicated to an extensive literature review. Section~\ref{sec:gaussian-framework} adapts the rate-distortion-framework to Gaussian signals. Section~\ref{numerical-evidence} presents the numerical results for both theoretical framework and real-world counterparts. In Section~\ref{sec:conclusion}, we draw the conclusions.

\section{Rate vs. distortion vs. distinguishability}
\label{sec:RDD}
Let us use $\x$ to represent a generic stationary discrete-time source that, at each time $t$ generates a random vector $\x[t]\in\RR^n$. This source mostly generates typical signal instances but rarely gives rise to anomalies. To distinguish between normal and anomalous behavior, we refer to two independent sources $\xok$ and $\xko$, respectively, where the two vectors obey two different probability density functions (PDFs) $\fok$ and $\fko$. At any time $t$, the observable process is either $\x[t]=\xok[t]$ or $\x[t]=\xko[t]$. Exploiting stationarity, from now on we drop the time index $t$.

\subsection{Rate vs. distortion}

The observable $\x$ is converted into symbols suitable for transmission over a limited bandwidth channel, which can only handle a maximum number of symbols per second. If this rate is insufficient, compression discards some of the information to match the capacity of the channel. More in detail, $\x\in\RR^n$ is encoded into a reduced information symbol $\y$ that can be transmitted through the channel and later expanded to $\xh$.
Consequently, the receiver works on a signal $\xh$ that diverges from the original, with distortion increasing as the channel capacity decreases. The relationship between rate and distortion is thoroughly discussed in rate-distortion theory \cite[Chapter 13]{Cover_1991} and can be used to design the encoder compressing $\x$.

Typically, since anomalies are rare, designers tune the encoder to match the normal behavior of the signal source for which $\x=\xok$ only. Hence, distortion is defined as
\begin{equation}
\label{eq:distortiondef}
D = \Exp\left[\left\|\xok-\xhok\right\|^2\right]
\end{equation}
where $\Exp[\cdot]$ stands for expectation.

Compression strength is measured by the rate of transmitting $\y$ that maps to $\xh$. Since the receiver mapping from $\y$ to $\xh$ is injective, the encoder supposes $\x=\xok$, such rate can be quantified by the mutual information between $\xok$ and $\xhok$ indicated as $I\left(\xhok;\xok\right)$ \cite[Chapter 8]{Cover_1991}.

With these definitions, the relationship between minimum achievable rate $\rho$ and maximal accepted distortion $\delta$ is regulated by \cite[Theorem 13.2.1]{Cover_1991}
\begin{equation}
\label{eq:rate-distortion}
\rho(\delta)=\inf_{\fxhgx} I\left(\xhok;\xok\right) \quad\text{s.t. $D\le \delta$}
\end{equation} 
where $\fxhgx$ is a conditional PDF modeling the possibly stochastic mapping characterizing the encoder-decoder pair so that the PDF of $\xh$ is
\begin{equation}
\label{eq:encoder}
\fxh(\alpha)=\int_{\RR^n}\fxhgx(\alpha,\beta)\fx(\beta)\dd\beta .
\end{equation}

Although \cite[Theorem 13.2.1]{Cover_1991} defines the rate-distortion function in the discrete case, it can also be proved for well-behaved continuous sources \cite[Chapter 13]{Cover_1991} as considered in this work.

\subsection{Distinguishability}

The encoder $\fxhgx$ designed considering $\x=\xok$ is used for normal and anomalous signals. Hence, the retrieved signal can be either $\xh=\xhok$ or $\xh=\xhko$, where the two vectors have in general different PDFs $\fokxh$ and $\fkoxh$ given by \eqref{eq:encoder} for $\fx=\fokx$ and for $\fx=\fkox$.

As compression increases (and thus the rate decreases), the statistics of the two vectors $\xhok$ and $\xhko$ tend to converge. In the limit, maximum compression (and thus zero rate) is achieved with a constant $\xh$ that is independent of the source. This reveals a second trade-off: the one between compression and distinguishability between two possible signal sources $\xhok$ and $\xhko$, which conditions the effectiveness of any detector acting on $\xh$ to detect possible anomalies.

In \cite{Marchioni_2024TSMC}, the authors define distinguishability in information-theoretic terms and under two different assumptions.
If the detector is only aware of the normal source, i.e., an anomaly-agnostic detector, then distinguishability can be quantified by
\begin{equation}
\label{eq:zeta}
Z = \left|\int_{\RR^n} \left[\fokxh(\alpha) - \fkoxh(\alpha)\right] \log_2\fokxh(\alpha) \dd\alpha \right|.
\end{equation}
where $Z$ is defined as $|\zeta|$ in \cite{Marchioni_2024TSMC}.

Conversely, when the detector knows the statistics of both normal and anomalous sources, i.e., an anomaly-aware detector, then distinguishability can be measured with 
\begin{eqnarray}
\label{eq:kappa}
J = \int_{\RR^n} \left[\fokxh(\alpha) - \fkoxh(\alpha)\right] \log_2\left[\frac{\fokxh(\alpha)}{\fkoxh(\alpha)}\right] \dd\alpha .
\end{eqnarray}

Note that $J$ (defined as $\kappa$ in \cite{Marchioni_2024TSMC}) is equivalent to a well-known divergence measure for binary classification and pattern recognition problems \cite{Jeffreys_RSPA1946,
Kullback_1951AMS,
Toussaint_TSMC1978,
Morgera_TPAMI1984}.

Distinguishability is reflected in the ability of a detector to determine the signal source. However, in practice, signals PDFs are hardly ever available or estimable with sufficient accuracy. For this reason, in practical cases, we measure the detector's performance in terms of probability of detection $\Pdet$ \cite{Marchioni_2024TSMC}
\begin{equation}
\Pdet = \begin{cases}
\AUC & \text{if $\AUC \ge 0.5$}\\
1 - \AUC & \text{if $\AUC < 0.5$}\\
\end{cases}
\end{equation}
with $\AUC$ indicating the Area-Under-the ROC Curve, defined starting from the anomaly score $s$ produced by the detector $\AUC={\rm Prob}\{s( \xhko ) > s( \xhok ) \}$. In contrast to the accuracy and similar metrics computed from binary labels, we prefer $\AUC$ and $\Pdet$ because they are directly obtained from the anomaly scores without needing any threshold that often depends on the specific application.

\subsection{Rate vs. distortion vs. distinguishability}

The figures of merit in \eqref{eq:zeta} and \eqref{eq:kappa} are used in \cite{Marchioni_2024TSMC} after \eqref{eq:rate-distortion} is solved to associate a distinguishability level to each possible pair of $\rho$ and $\delta$ optimally addressing the classical rate-distortion trade-off. The authors also show that, in the case of Gaussian sources, the information-theoretical measures anticipate the performance expressed in terms of $P_{\rm det}$.

Here, we address the more general trade-off between rate-distortion and distinguishability. To measure the latter we rely on \eqref{eq:zeta} and \eqref{eq:kappa} and model the triple trade-off by the following optimization problem:
\begin{equation}
\label{eq:rate-distortion-distinguishability}
\rho(\delta, \omega)=\inf_{\fxhgx} I\left(\xhok;\xok\right)\;\; \text{s.t.}\;\;
\begin{array}{c}
D\le \delta \wedge \Exp[Z]\ge\omega \\
\text{or} \\
D\le \delta \wedge J\ge\omega
\end{array}
\end{equation} 
\noindent where $\omega$ is the minimum tolerated level of distinguishability and only one of the two alternative constraints is considered.

In the first distinguishability constraint, $\Exp[Z]$ denotes the average performance quantified by $Z$ across a set of anomalies. We bound the expected value rather than $Z$ directly because evaluating the latter requires prior knowledge of $\fkox$, which is not available for the anomaly-agnostic detector.

The relationship between $\rho$, $\delta$, and $\omega$ will draw a Pareto efficiency surface in the three-dimensional parameter space, i.e., the locus of points moving from which at least one of the three merit figures becomes suboptimal.

Note that in both \eqref{eq:rate-distortion} and \eqref{eq:rate-distortion-distinguishability} we consider the rate and distortion corresponding to the $\ok$ source since the task is anomaly detection. Anomaly detection is intended as a one-class classification where the classifier distinguishes normal inputs from potential anomalies that are assumed to be rare events. This scenario is different from binary classification in which the two classes have the same relevance and there would be no reason to tune the rate and distortion on only one of the two classes.

\section{Related Works}
\label{sec:SOA}
The single trade-offs between rate-distinguishability, rate-distortion, and distortion-distinguishability have already been analyzed in the literature. However, the joint trade-off rate-distortion-distinguishability has been explored only for classification, but never for anomaly detection.

\subsection{Rate-distinguishability}
The impact of compression on the ability of a detector to differentiate between two sources of information has been first studied in \cite{Ahlswede_TIT1986}. In this work, the issue of hypothesis testing is addressed within the context of a single information source subjected to a rate constraint. However, the authors do not consider any constraint on distortion as the original signals are not required to be recovered.

In \cite{Tishby_Allerton1999}, the authors present a framework, named the information bottleneck, that allows to leverage the trade-off between rate and a broadly defined distortion criterion, not necessarily targeting reconstruction as in \eqref{eq:distortiondef}. This criterion determines which aspects of the information content within the original signal should be retained during compression, causing the compressor to restrict the rate while preserving specific features. In particular, the preserved features encompass the information present in the original signal common with a second signal introduced to suit the needs of the target application. For instance, this adaptation of compression has been applied to anomaly detection \cite{Crammer_ICML2004, Crammer_ICML2008}. Although the distortion criterion can be generalized to anomaly detection, such a framework only handles a single trade-off with rate and cannot describe the relation between distortion and distinguishability.
For a complete review of the information bottleneck, see \cite{Hu_TPAMI2024}.

Similarly, the authors in \cite{Singh_ICIP2020} proposed an end-to-end framework that jointly optimizes the rate considering a specific task objective. They then evaluated the method's effectiveness in the context of classification.

\subsection{Rate-distortion}

Lossy compression schemes are often designed to minimize the rate given a distortion constraint. This trade-off has been extensively addressed in the field of Information Theory \cite{Cover_1991}, which has brought to the definition of an optimal rate-distortion curve. This curve represents the theoretical lower bound for the transmission rate while maintaining a specified maximum distortion level.
Such a curve is source-dependent and can be explicitly formulated for only a few sources, such as Gaussian sources. Moreover, in practical scenarios, the optimal rate-distortion curve is not attainable. However, it still serves as a valid benchmark for evaluating the performance of a compression method, even when approached suboptimally. 

Among the most adopted compression techniques it is worth mentioning transform coding \cite{Goyal_SPM2001} that, due to its computational efficiency, has been the core of the image compression standards like JPEG \cite{Wallace_TCE1992} and JPEG 2000 \cite{Skodras_SPM2001}. Another important class of compressors is based on vector quantization \cite{Gray_1984ASSP}, a technique that, by simultaneously discretizing all components of a multivariate source, allows to achieve compression performance close to the theoretical rate-distortion bounds \cite{Gray_TIT1998}.
For a more comprehensive explanation of lossy compression and the rate-distortion trade-off, refer to \cite{Gray_TIT1998}.

In addition, with the development of deep learning, end-to-end optimized image compression frameworks are gaining popularity \cite{
Theis_2017ICLR,
Mentzer_CVF2018,
Hu_TPAMI2022,
Duan_2024TPAMI}. Their success is due to their ability to simultaneously optimize all the modules forming a compressor, leading to improved performance compared to traditional methods. The study in \cite{Hu_TPAMI2022} presents a comprehensive survey and benchmark of end-to-end learned image compression methods.

\subsection{Distortion-distinguishability}

The well-established rate-distortion framework has been extended in \cite{Marchioni_2024TSMC} to consider the scenario in which the main goal of compression is still to minimize distortion given a constrained rate, but there may be secondary tasks such as anomaly detection. They define two metrics $Z$ and $J$ to measure distinguishability and then study the trade-off between distortion and detection performance when the signal is optimally compressed in rate-distorion terms. 

On the other hand, some works focus on practical lossy compressors that can administer the distortion-classification trade-off \cite{
Oehler_1995TPAMI,
Baras_1999TIT,
Gupta_2003TIT,
Lexa_2012TSP}. These works focus on vector quantizers and tailor them for the dual purpose of classification and compression. Specifically, they enhance standard vector quantizer design algorithms by integrating a relevant term into the distortion measure, addressing classification performance.

Closely related to lossy compression is dimensionality reduction which aims to approximate the signal with a lower-dimensional representation \cite{
Kramer_AICE1991,
Vetterli_SPM2001,
Jolliffe_2002,
Rezende_2014PMLR}. In principle, dimensionality reduction algorithms cannot be considered proper lossy compressors since they do not include a quantization stage that limits the rate. Nevertheless, by seeking a more compact representation, dimensionality reduction does lead to a loss of information and distorts the signal. For this reason, a dimensionality reduction stage can be included in a proper lossy compressor to achieve state-of-the-art performance \cite{
Theis_2017ICLR,
Mentzer_CVF2018,
Cheng_PCS2018,
Hu_TPAMI2022}.
Moreover, dimensionality reduction is also exploited in anomaly detection \cite{Ruff_PMLR2018, Ruff_2021ProcIEEE, Seonho_2022AMAI} to mitigate the problem of the curse of dimensionality and increase the signal-to-noise ratio. Indeed, a detector is often more effective when it works on a lower dimensional space compared to the original domain.

The authors in \cite{Cao_2019TC} proposed a dimensionality reduction scheme based on an autoencoder (AE) to build an anomaly detector. This scheme incorporates a well-suited regularization term enhancing the performance of the detector that operates on the compressed signal. 
Similarly, the authors in \cite{Enttsel_AICAS2024} leverage $Z$ metric to identify an alternative regularization term to train the AE and enhance detection performance. They also analyzed the influence of this term on the balance between distortion and distinguishability. 
These two regularization approaches are then compared in \cite{Enttsel_EUSIPCO2024} for the detection of different anomalous sources.

In the context of hyperspectral imagery, the work in \cite{Fowler_TIP2012} investigates the impact of random projection-based dimensionality reduction on anomaly detection and reconstruction by adapting the Reed-Xiaoli algorithm \cite{Reed_TSP1990} to the compressed domain and providing two reconstruction paths for normal and anomalous pixels.
In the same context, the authors in \cite{Du_2007GRSL} assess a compression scheme that combines Principal Component Analysis (PCA) and JPEG 2000 by analyzing the rate-distortion trade-off and the impact of the information loss on the anomaly detection task.

\subsection{Rate-distortion-accuracy}

The three-fold trade-off between rate, distortion, and distinguishability has already been investigated in the case of classification. 
In \cite{Yanting_2005SJ}, the authors developed a rate-distortion framework to examine the performance of different combinations of compression schemes and algorithms that estimate the orientation of the target. 
The authors in \cite{Zhang_2023DSP} analyze the joint rate-distortion-accuracy trade-off with a theoretical framework from which they provide insight for practical compression schemes specifically for the classification of compressed images.
The same triple trade-off is addressed in \cite{Luo_2021DCC, Chamain_2022IOTJ, Bai_2022AAAI} with practical compression schemes. In \cite{Luo_2021DCC}, the authors tune the JPEG quantization tables to enhance either the classification or reconstruction performance given a certain rate. 
While in \cite{Chamain_2022IOTJ,Bai_2022AAAI}, the authors introduce a Neural Network-based frameworks that allow to jointly optimize a compressor, a classifier, and possibly a decoder.

\subsection{Rate-distortion-third metric}

The aforementioned works \cite{Luo_2021DCC, Chamain_2022IOTJ} belong to a broader paradigm that focuses on communication for joint human and machine vision and addresses the threefold trade-off between rate-distortion and a third metric that measures performance in an image/video task. In \cite{Torfason_2018ICLR, Liu_ICPR2022, Liu_2024ITJ}, the authors tackle the tasks of classification, semantic segmentation, object detection, foreground extraction, or depth estimation in images. In contrast, in \cite{Duan_2020TIP, Sheng_2024TPAMI, Yang_TPAMI2024}, the authors consider video processing tasks such as action recognition, denoising, super-resolution, scene classification, semantic segmentation, or object classification.

It is also worth mentioning the works \cite{
Blau_CVPR2018,
Blau_ICML2019} in which the authors study how the perception quality of the image is affected by distortion and then analyze a joint trade-off between rate, distortion, and a perception measure.

Regarding these related works, our research focuses on assessing the trade-off between rate, distortion, and the efficacy of a detector in distinguishing between a normal and an anomalous source in both supervised and unsupervised scenarios. Unlike traditional binary classification, anomaly detection assumes samples of the positive class (anomalies) as rare events leading the compression mechanism to be only adapted to the negative class (normal signal). This scenario, to our knowledge, has not been explored previously.

\section{Gaussian framework}
\label{sec:gaussian-framework}

The difference between \eqref{eq:rate-distortion-distinguishability} and \eqref{eq:rate-distortion} consists of the additional constraint that allows the encoder to be tuned to increase distinguishability. Solving \eqref{eq:rate-distortion-distinguishability} for generic sources is extremely difficult -- certainly more difficult than solving \eqref{eq:rate-distortion}.  For this reason, we address the problem by relying on some assumptions. 

Using $\Gauss{\vec \mu}{\cov}$ to indicate a multivariate Gaussian distribution with an average vector $\vec{\mu} \in \RR^{n}$ and a covariance matrix $\cov \in \RR^{n\times n} $, we consider $\xok\sim\Gauss{\mok}{\covok}$ and $\xko\sim\Gauss{\mko}{\covko}$  with,  in general, $\mok\neq\mko$ and $\covok\neq\covko$.

As a reference scenario, we assume $\mok=\mko=\vec{0}$, $\tr(\covok) = n$, and $\tr(\covko) = \alpha n$, where $\tr(\cdot)$ stands for matrix trace. 
In this setting, on average, each sample contributes with energy $1$ or $\alpha$ for normal and anomalous vectors, respectively.
Since there always exists an orthonormal transformation that makes $\covok$ diagonal, we consider $\covok={\rm diag}\left(\lambdavok \right)$, with $\lambdavok=\left(\lambdaok_0,\dots,\lambdaok_{n-1}\right)$ and $\lambdaok_0\ge\lambdaok_1\ge\dots\ge\lambdaok_{n-1}\ge 0$.

With such assumptions, from \cite{Marchioni_2024TSMC} it is known that the average anomaly is white, i.e., $\Exp[\covko] = \alpha \mat{I}_n = \white$. Moreover, the average distinguishability $\Exp[Z]$ and $\Exp[J]$ are related to those computed for the white anomaly $Z_{\white}$ and $J_{\white}$. In detail, according to Jensen's inequality $Z_{\white} \le \Exp[Z]$ while $J_{\white} \le \Exp[J]$ as shown in \cite{Marchioni_2024TSMC}. As a result, in \eqref{eq:rate-distortion-distinguishability} we can rely on  $Z_{\white}$ as a lower bound of $\Exp[Z]$.

For the constraint on $J$ in \eqref{eq:rate-distortion-distinguishability}, it is possible to solve that problem within the considered Gaussian framework by focusing on the anomalies for which $\covko={\rm diag}\left( \lambdavko \right)$. Here we consider
$\lambdavko=\left(\lambdako_0,\dots,\lambdako_{n-1}\right)$.

Though we will achieve a wider generality, we initially assume that encoding is {\em Gaussian-additive}, i.e., that $\fxhgx$ is such that $\Deltav=\xh-\x$ is a random vector made of zero-mean components that are independent of each other but jointly Gaussian with those of $\x$ (and thus $\xh$), where the  covariance matrices of all triads $\hat{x}^{\ok}_j,
x^{\ok}_j,
\Delta_j$ are
\begin{equation}
\label{eq:gaussian-additive-encoding}
\cov_{
\left(\hat{x}^{\ok}_j,
x^{\ok}_j,
\Delta_j\right)
}=
\begin{pmatrix}
\lambdaok_j-\theta_j & \lambdaok_j-\theta_j & 0 \\
\lambdaok_j-\theta_j & \lambdaok_j & -\theta_j\\
0 & -\theta_j & \theta_j
\end{pmatrix}
\end{equation}
\noindent for some variances $0\le\theta_j\le\lambdaok_j$ for $j=0,\dots,n-1$ that are the degrees of freedom defining the encoding.

For Gaussian-additive encodings, we may give two lemmas. 

\begin{lemma}
\label{lem:Gaussadd}
If $\xok$ is made of independent Gaussian components with zero mean and variance $\lambdaok_j$, then applying a
Gaussian-additive encoding with variances $\theta_j$ for $j=0,\dots,n-1$ yields the rate 
\begin{equation}
I\left(\xhok;\xok\right) = \frac{1}{2}\sum_{j=0}^{n-1}\log_2\frac{\lambdaok_j}{\theta_j}
\end{equation}
\noindent that is the minimum achievable by any encoding such that $\Exp\left[\left(\xhok_j-\xok_j\right)^2\right]=\theta_j$ for $j=0,\dots,n-1$.
\end{lemma}

\begin{proof}[Proof of Lemma \ref{lem:Gaussadd}]
Indicate with $H(\cdot)$ the differential entropy of a random vector, with $H(\cdot|\cdot)$ the conditional differential entropy between two random vectors, and replicate the well-known chain
\begin{align}
I(\xhok;\xok)
& = H\left(\xok\right) - H\left(\xok|\xhok\right)\\
& = H\left(\xok\right) - H\left(\xok-\xhok|\xhok\right)\\
& \ge H\left(\xok\right) - H\left(\xok-\xhok\right)
\end{align}
\noindent where the inequality depends on the fact that conditioning decreases entropy.

Set $\Deltav=\xhok-\xok$ and loosen the bound by noting that if $\cov^{\Deltav}$ is the covariance matrix of $\Deltav$, then $H\left(-\Deltav\right)$ is maximized when $\Deltav$ is a zero mean Gaussian vector with such a covariance. This reduces the above bound to
\begin{equation}
I\left(\xhok;\xok\right)\ge H\left(\xok\right)-\frac{n}{2}\log_2\left(2\pi e\right)-\frac{1}{2}\log_2\left|\cov^{\Deltav}\right|  
\end{equation}

To make the bound even looser, we may now notice that the variances $\theta_j$ of the single components of $\Deltav$ are given, and also the diagonal elements of $\cov^{\Deltav}$ are given. Hence, by Hadamard's inequality, the determinant $\left|\cov^\Delta\right|$ is maximum when all other entries of $\cov^{\Deltav}$ are zeros, and thus the components of $\Deltav$ are independent of each other.

All the conditions given above to yield the minimum possible lower bound on the rate, are met by a Gaussian additive encoding with those $\theta_j$. Hence, such an encoding yields the minimum possible rate that is
\begin{equation}
H\left(\xok\right)-H\left(-\Deltav\right)=
\frac{1}{2}\sum_{j=0}^{n-1}\log_2\frac{\lambdaok_j}{\theta_j} 
\end{equation}

By applying such an encoding, we have $\Exp\left[\left(\xhok_j-\xok_j\right)^2\right]=\theta_j$ for $j=0,\dots,n-1$ and thus
\begin{equation}
D=\sum_{j=0}^{n-1}\theta_j
\end{equation}
\end{proof}

\begin{lemma}
\label{lem:GDzk}
For a Gaussian-additive encoder with parameters $\theta_j$ for $j=0,\dots,n-1$, we have
\begin{align}
\label{eq:GD}
D & = n-\sum_{j=0}^{n-1}\lambdaok_j\xi_j\\
\label{eq:Gz}
Z & = \frac{1}{2\ln 2}\left|\sum_{j=0}^{n-1}\stack_j\xi_j\right| 
\\
\label{eq:Gk}
J &=
\frac{1}{2\ln 2}
\sum_{j=0}^{n-1}
\frac{\stack_j^2\xi_j^2}{1-\stack_j\xi_j}
\end{align}
\noindent where
$\stack_j=1-\nicefrac{\lambdako_j}{\lambdaok_j}$ 
and
$\theta_j=\lambdaok_j\left(1-\xi_j\right)$
so that $\xi_j\in[0,1]$ for $j=0,\dots,n-1$ are the normalized degrees of freedom defining the encoder.

Both $Z$ and $J$ are convex functions of the $\xi_j$.
\end{lemma}

\begin{proof}[Proof of Lemma \ref{lem:GDzk}]
If we assume Gaussian-additive encoding, we have that $\xhok=\xok+ \Deltav$ is made of independent zero-mean Gaussian components with variance $\sigma^2_{\xhok_j}=\lambdaok_j-\theta_j$ for $j=0,\dots,n-1$.
The results of \cite{Marchioni_2024TSMC} can be straightforwardly applied to this slightly generalized context, yielding the encoding $\fxhgx$ \cite[Lemma 1]{Marchioni_2024TSMC} and its application to the components of $\xko$ \cite[Lemma 2]{Marchioni_2024TSMC}.

This yields that $\xhko$ is made of independent, zero-mean, Gaussian components with variances
\begin{equation}
\sigma^2_{\xhko_j}=\left(1-\frac{\theta_j}{\lambdaok_j}\right)
\left[\lambdako_j\left(1-\frac{\theta_j}{\lambdaok_j}\right) + \theta_j\right]
\end{equation}
\noindent for $j=0,\dots,n-1$.

Finally, \cite[Lemma 3]{Marchioni_2024TSMC} allows to compute both $Z$ and $J$ that turn to be a function of the ratios 
\begin{equation}
    u_j=\frac{\sigma^2_{\xhko_j}}{\sigma^2_{\xhok_j}}=1-\stack_j\xi_j
\end{equation}

In particular, the application of \cite[Lemma 3]{Marchioni_2024TSMC} to the definitions of $Z$ and $J$ leads to \eqref{eq:Gz} and \eqref{eq:Gk}.

Both $Z(\xiv)$ and $J(\xiv)$ are convex functions of $\xiv=(\xi_0,\dots,\xi_{n-1})$. In fact, $Z$ is linear for $\sum_{j=0}^{n-1}\stack_j\xi_j\ge 0$ and for $\sum_{j=0}^{n-1}\stack_j\xi_j< 0$. Moreover, since $\stack_j\xi_j < 1$, $\nicefrac{\left(\stack_j\xi_j\right)^2}{1-\stack_j\xi_j}$ is a convex function for $j = 0, \dots, n-1$. Hence, $J$, being a sum of convex functions, is convex.


In the new degrees of freedom, distortion is
\begin{equation}
D=\sum_{j=0}^{n-1}\theta_j=\sum_{j=0}^{n-1}\lambdaok_j(1-\xi_j)=n-\sum_{j=0}^{n-1}\lambdaok_j\xi_j
\end{equation}
\end{proof}

Assume now that a distortion and a distinguishability constraint are given.
Lemma \ref{lem:Gaussadd} says that starting from a feasible encoding $\fxhgx'$, one may compute $\theta_j=\Exp\left[\left(\xhok_j-\xok_j\right)^2\right]$ and then consider the Gaussian-additive encoder $\fxhgx''$ with those parameters.
If $\fxhgx''$ satisfies both constraints, then it yields a rate that is not greater than the one yielded by $\fxhgx'$.
Yet, since $D=\sum_{j=0}^{n-1}\theta_j$ independently of the specific underlying encoder, the only constraint one has to re-check when moving from $\fxhgx'$ to $\fxhgx''$ is the distinguishability one.

To formalize these considerations, we use \eqref{eq:Gz}, \eqref{eq:Gk}, and \eqref{eq:GD} and define the two subsets of $[0,n]\times\RR^+$ 
\begin{align}
\label{eq:Hz}
H_Z & = \left\{
(\delta,\omega) | \exists \xi_0,\dots, \xi_{n-1}\;\; \eqref{eq:GD}=\delta \wedge
\eqref{eq:Gz}=\omega
\right\}\\
\label{eq:Hk}
H_J & = \left\{
(\delta,\omega) | \exists \xi_0,\dots, \xi_{n-1}\;\; \eqref{eq:GD}=\delta \wedge
\eqref{eq:Gk}=\omega
\right\}.
\end{align}

Due to its special features, when $(\delta,\omega)\in H_{Z|J}$, the corresponding rate-distortion-distinguishability problem \eqref{eq:rate-distortion-distinguishability} is solved by a Gaussian-additive encoding.
Out of those domains, such a property cannot be derived but must be listed as an assumption along with the Gaussian nature of the normal and anomalous sources.

Whatever the level of generality of the derivations that follow, the RDD trade-off we are investigating depends on the solution of \eqref{eq:rate-distortion-distinguishability} rewritten in terms of the normalized degrees of freedom $\xi_0,\dots,\xi_{n-1}$, i.e., 
\begin{eqnarray}
\label{eq:q}
\rho =	& \min_{\xi_0,\dots,\xi_{n-1}} & -\frac{1}{2\ln 2}\sum_{j=0}^{n-1}\ln\left(1-\xi_j\right)\\
\label{eq:c0}
	& \text{s.t.} & 0\le \xi_j \le 1\\
\label{eq:c1}
& \text{s.t.} & \sum_{j=0}^{n-1}\lambdaok_j\xi_j\ge n-\delta\\
\label{eq:c2-a}	
\text{either}
	& \text{s.t.} & 
	\left|\sum_{j=0}^{n-1}\stack^{\white}_j\xi_j\right|\ge 2\dist \ln 2\\
\label{eq:c2-b}
\text{or}
& \text{s.t.} & \sum_{j=0}^{n-1}\frac{\left(\stack_j\xi_j\right)^2}{1-\stack_j\xi_j}\ge 2\dist \ln 2
\end{eqnarray}
where we should consider only one of the two constraints \eqref{eq:c2-a} and \eqref{eq:c2-b} at a time and with $\stack^{\white}_j=1-\nicefrac{\alpha}{\lambdaok_j}$ in case of \eqref{eq:c2-a}. Note that to derive \eqref{eq:c2-a} the previously described relation $Z_{\white} \leq \Exp[Z]$ has been exploited. 
Table~\ref{tab:assumptions} summarizes the assumptions under which the rate-distortion-distinguishability problem in \eqref{eq:q}-\eqref{eq:c2-b} is valid. 
The assumption on the normal signal does not bring a loss of generality. In the anomaly-agnostic case, we only assume the knowledge of scaling parameter $\alpha$ that sets the average energy of the anomaly w.r.t. that of the normal signal. In the anomaly-aware case, the considered anomalies are characterized by a diagonal covariance matrix.

\begin{table}[]
    \centering
    \caption{Assumptions on the signals and the encoder on which the problem in \eqref{eq:q}-\eqref{eq:c2-b} relies.}
    \begin{tabular}{lllll}
        \toprule
        \multirow{2}{*}{} \multirow{1}{*}{}
        Signals & & $\muv$ & $\cov$ & $\tr(\cov)$ \\
        $\vec{x} \sim \Gauss{\muv}{\cov}$ & & & & \\
        \midrule
        $\ok$ & & $\vec{0}$ & $\diag\left( \lambdavok \right)$ & $n$ \\
        \multirow{2}{*}{$\ko$} &
        RDD Z
        & $\vec{0}$ & $\covko$ & $\alpha n$ \\
         & 
         RDD J
         & $\vec{0}$ & $\diag\left( \lambdavko\right)$ & $\alpha n$ \\
         \midrule \midrule
        Encoder & \multicolumn{4}{c}{Gaussian-additive if $(\delta,\omega)\notin H_{Z|J}$}\\
        \bottomrule
    \end{tabular}
    \label{tab:assumptions}
\end{table}

\subsection{Solution to the optimization problems}

If we concentrate on \eqref{eq:c2-a}, we have that the constraint can be expressed as a disjunction
\begin{equation}
\sum_{j=0}^{n-1}\stack^{\white}_j\xi_j \ge  2\dist \ln 2
\;\;\;\cup\;\;\;
\sum_{j=0}^{n-1}\stack^{\white}_j\xi_j \le -2\dist \ln 2
\end{equation}

The disjoint form allows us to find the minimum rate solution by solving two separate maximization problems in which the merit function is concave, the constraints are linear, and the feasibility space is convex.

As already proven, the left-hand side in \eqref{eq:c2-b} is a convex expression that must be paired with the $\leq$ relation symbol to satisfy the convex programming rules. In our case, the presence of $\geq$ relation symbol makes the problem a so-called reverse convex optimization problem \cite{Hillestad_1980AMO}. This problem can be solved heuristically by the algorithm proposed in \cite[Algorithm 5.1]{Tuy_JOTA1987}. With this approach, the problem is reduced into a sequence of sub-problems each of which consists in maximizing a convex function. We rely on the Disciplined Convex-Concave Programming (DCCP) framework \cite{Shen_2016CDC} to solve each subproblem.

\captionsetup[subfigure]{width=0.9\columnwidth}
\begin{figure*}
\centering
  \subfloat[Rate-distinguishability and rate-distortion trade-offs in the anomaly-agnostic case with $\alpha=1$.]{\includegraphics[]{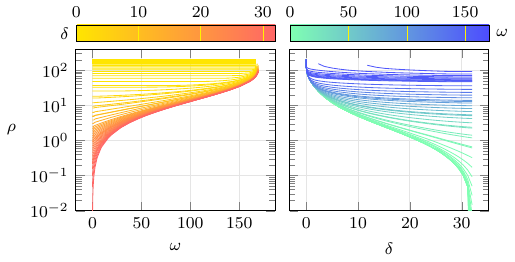} \label{fig:RDD-tradeoffs-unsupervised-white-omega}}
  \subfloat[Rate-distinguishability and rate-distortion trade-offs in the anomaly-aware case with $\covko=\mat{I}_n$.]{\includegraphics[]{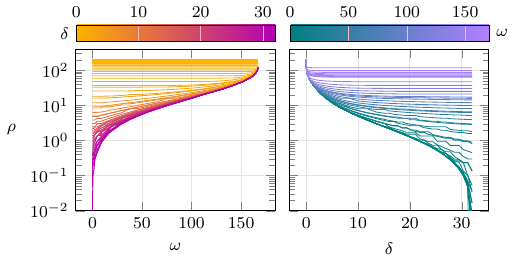} \label{fig:RDD-tradeoffs-supervised-white-omega}}
\caption{Solution to the RDD Pareto problems.}
\label{fig:RDD-cuts-white-omega}
\end{figure*}

\section{Numerical Evidence}
\label{numerical-evidence}

\captionsetup[subfigure]{width=0.45\columnwidth}
\begin{figure}
\centering
  \subfloat[Pareto surface in the anomaly-agnostic case with $\alpha=1$.]{\includegraphics[width=0.42\columnwidth]{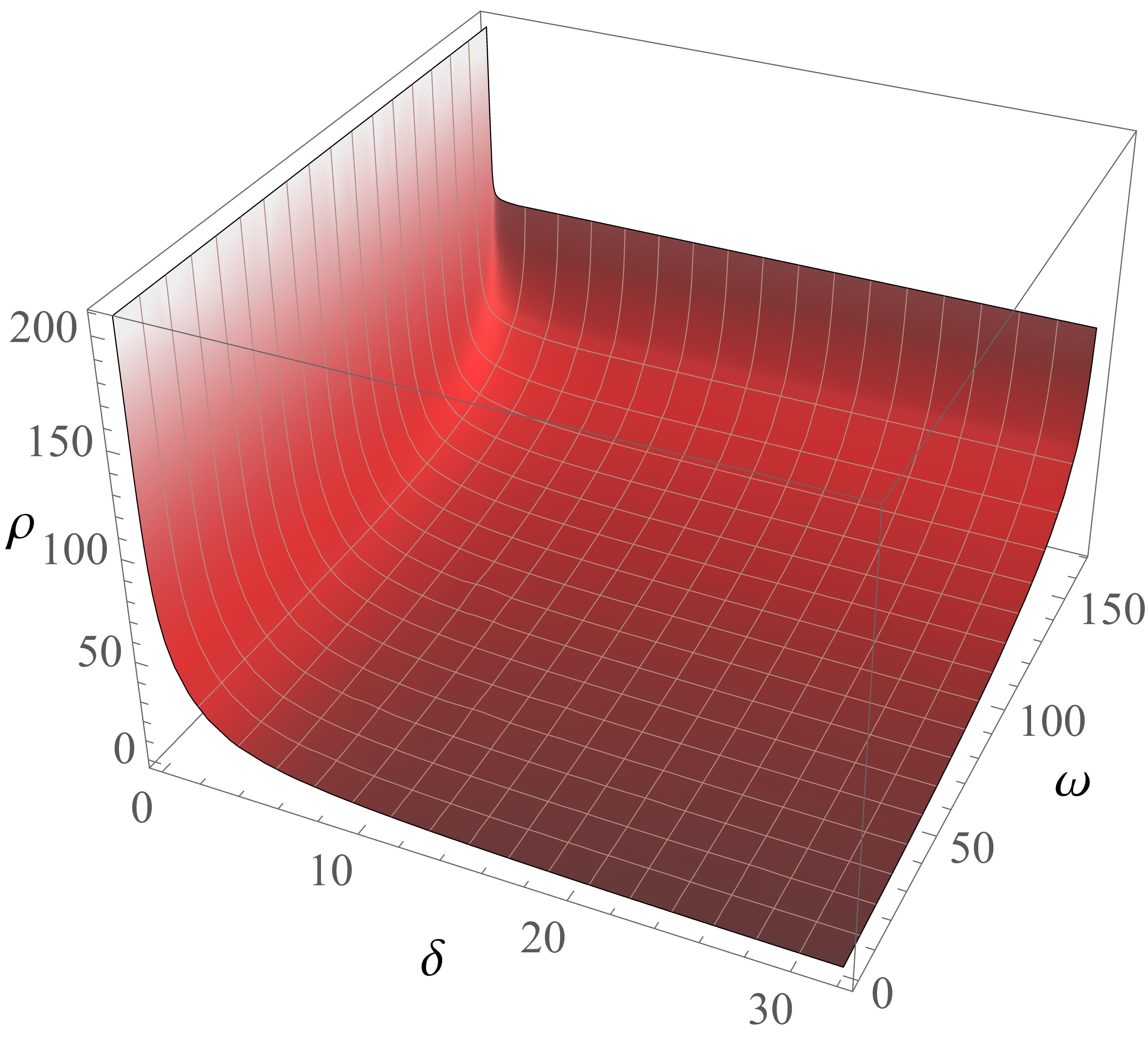} \label{RDD-surface-unsupervised-white-omega}}\hspace{3mm}
  \subfloat[Pareto surface in the anomaly-aware case with $\covko=\mat{I}_n$.]{\includegraphics[width=0.42\columnwidth]{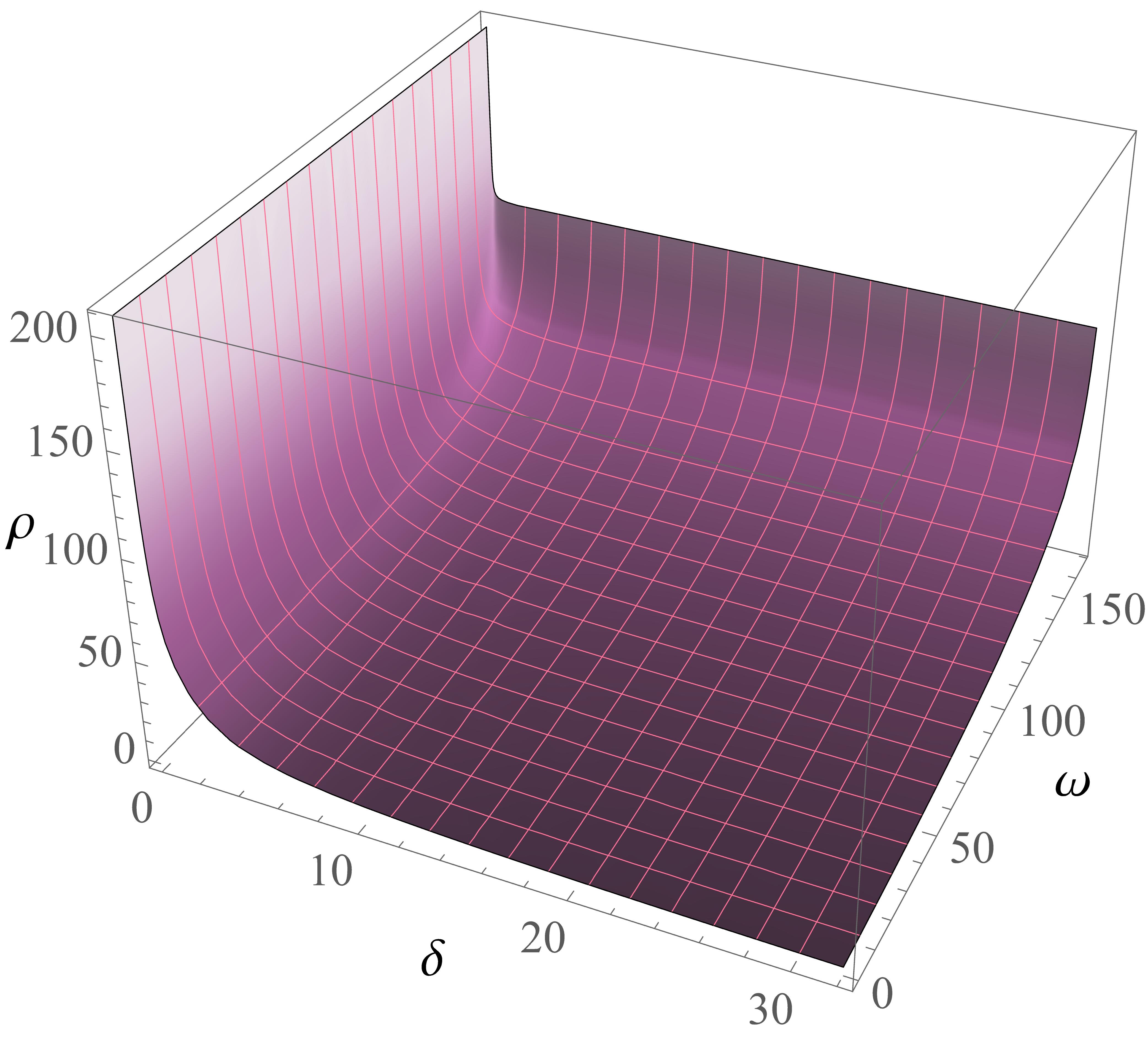} \label{RDD-surface-supervised-white-omega}}
\caption{Rate-distortion-distinguishability Pareto surfaces.}
\label{fig:RDD-surface-white-omega}
\end{figure}

\captionsetup[subfigure]{width=0.9\columnwidth}
\begin{figure*}[t!]
\centering
  \subfloat[Rate-distinguishability and rate-distortion trade-offs for the LD.]{\includegraphics[]{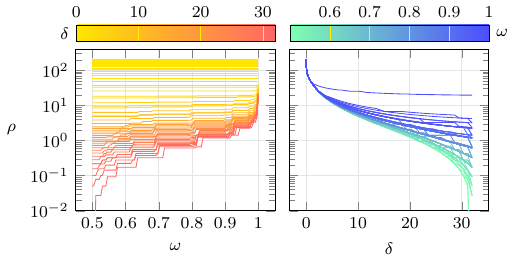} \label{fig:RDD-cuts-unsupervised-white-PD}}
  \subfloat[Rate-distinguishability and rate-distortion trade-offs for the NPD.]{\includegraphics[]{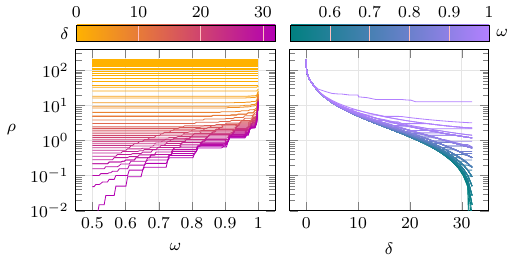} \label{fig:RDD-cuts-supervised-white-PD}}
\caption{Rate-distortion-distinguishabilty trade-off for the white anomaly and realistic detectors.}
\label{fig:RDD-white-PD}
\end{figure*}

\begin{figure}
    \centering
    \includegraphics[]{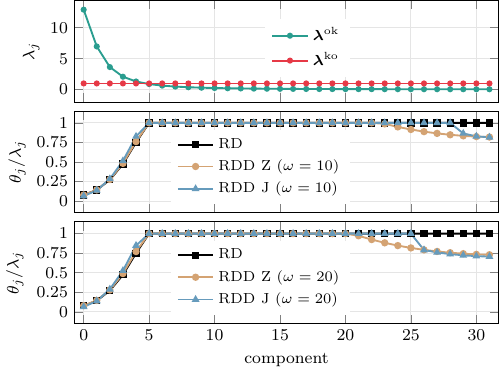}
    \caption{Signal components relative distortion profiles for the rate-distortion (RD) and rate-distortion-distinguishability (RDD) problems with $\delta=10$.} 
    \label{fig:RDD-profiles}
\end{figure}

The numerical evidence consists of four steps. We start showing and discussing the solution of the optimization problem in \eqref{eq:q}-\eqref{eq:c2-b} that assumes involved signals as Gaussian and the encoder as Gaussian-additive (see Table~\ref{tab:assumptions}).
Then, we progressively relax those assumptions with three examples:
\begin{enumerate*}[label=(\roman*), font=\itshape]
    \item a compression scheme based on dimensionality reduction that allows us to deal with Gaussian signals;
    \item a variation of the JPEG compression standard in which we can exploit \eqref{eq:q}-\eqref{eq:c2-b} to optimize distortion and detection jointly;
    \item a compression scheme based on an autoencoder that exploits a regularization term to enhance detection performance.
\end{enumerate*}

\subsection{Solution to the Pareto problems}

First, we present the solutions of both RDD Pareto problems where the setup for signal generation is identical to \cite{Marchioni_2024TSMC}. In particular, $n=32$, $\covok$ is a diagonal matrix with the $\lambdavok$ profile displayed in Fig.~\ref{fig:RDD-profiles} (top plot). The solution to the RDD minimization problem in \eqref{eq:q}-\eqref{eq:c2-b} by considering \eqref{eq:c2-a}, i.e., the anomaly-agnostic scenario, is done by selecting $\alpha=1$ while the anomaly-aware case, i.e., by considering \eqref{eq:c2-b} as a detectability constraint, uses $\covko=\mat{I}_n$ (see Fig.~\ref{fig:RDD-profiles} for the profiles of $\lambdavko$). 

The solutions to RDD Pareto problems are presented in Fig.~\ref{fig:RDD-tradeoffs-unsupervised-white-omega} and Fig.~\ref{fig:RDD-tradeoffs-supervised-white-omega} for \eqref{eq:c2-a} and \eqref{eq:c2-b} constraint choice, respectively.
 These figures display distinguishability-rate curves for different distortion levels and distortion-rate curves for different distinguishability levels. In both anomaly-aware and anomaly-agnostic scenarios, they highlight that, for a constant rate, the distinguishability increases with increasing distortion.
Meanwhile, given a distortion, distinguishability improves with higher rates. We also present the Pareto surface in Fig.~\ref{fig:RDD-surface-white-omega} to further highlight these trends.

Since $Z$ and $J$ cannot be directly compared, we measure the distinguishability performance in terms of $P_{\rm det}$ to compare anomaly-agnostic and anomaly-aware scenarios.
To estimate $P_{\rm det}$, we first generate $N = 10^4$ examples of $\xhok$ and $\xhko$. Then, for each $i$-th example, we compute the anomaly score with the likelihood-based detector (LD) $s(\xh_i) = -\log \fokxh(\xh_i)$ for the anomaly-agnostic case and the Neyman-Pearson detector (NPD) $s(\xh_i)=\log\fkoxh(\xh_i)-\log\fokxh(\xh_i)$ for the anomaly-aware scenario. Finally, normal and anomalous scores allow us to estimate $\AUC$ \cite{Fawcett_PATREC2006}. 
The results are shown in Fig.~\ref{fig:RDD-white-PD} for the agnostic and aware scenarios. It is evident how, in the anomaly-aware scenario, the same level of distinguishability, measured in terms of $P_{\rm det}$, can be obtained with a lower rate. This is expected since the detector can leverage the side information on the anomaly.

While Fig.~\ref{fig:RDD-white-PD} compares performances, we can visualize how the optimization problem in \eqref{eq:q}-\eqref{eq:c2-b} affects the compression mechanism in Fig.~\ref{fig:RDD-profiles}. As mentioned before, the top plot shows the eigenvalues profiles of the normal ($\lambdavok$) and anomalous ($\lambdavko$) signals, while the two plots below display the relative distortion the compressor applied to the input components depending on the active constraints in the RDD problem for a distortion level of $\delta=10$ and two different values of the distinguishability constraint $\omega$. We refer to RD when only \eqref{eq:c0} and \eqref{eq:c1} are active, i.e., $\omega=0$, which is equivalent to the traditional rate-distortion problem in \eqref{eq:rate-distortion} with Gaussian sources whose solution is presented in \cite[Theorem 13.3.3]{Cover_1991}. Instead, we refer to RDD Z and RDD J when also \eqref{eq:c2-a} or \eqref{eq:c2-b} is active, i.e., the anomaly-agnostic and anomaly-aware scenarios, respectively.
In the RD case, the compressor distorts the components corresponding to the largest eigenvalues in $\lambdavok$ (principal components) of a fixed amount and completely distorts the remaining minor components. In other words, the compressor sacrifices the minor components to allocate more rate to those carrying more information. As a result, the relative distortion $\theta_j/\lambda_j$ increases with $j$ up to reach $1$, i.e., complete distortion. The RDD cases are similar except that the compressor frees some rate from the principal components to allocate it to the components with the lowest energy. This phenomenon is more evident when the distinguishability constraint is stricter, i.e., a higher $\omega$. The difference between RDD Z and RDD J is that in the latter case, the compressor can leverage the knowledge of the anomaly distribution to allocate some rate where $\lambdavok$ and $\lambdavko$ differ most.

The fact that the components with the lowest variance are the most informative for detection is consistent with the theoretical results presented in \cite{Tax_ICONIP2003}, where the authors have shown that retaining the components with the lowest variance during dimensionality reduction is more beneficial for detection than retaining the major ones. This has been empirically confirmed for image data \cite{Rippel_TIM2021} and time series \cite{Marchioni_ITJ2020}.

\begin{figure}
\centering
    \includegraphics[scale=0.3]{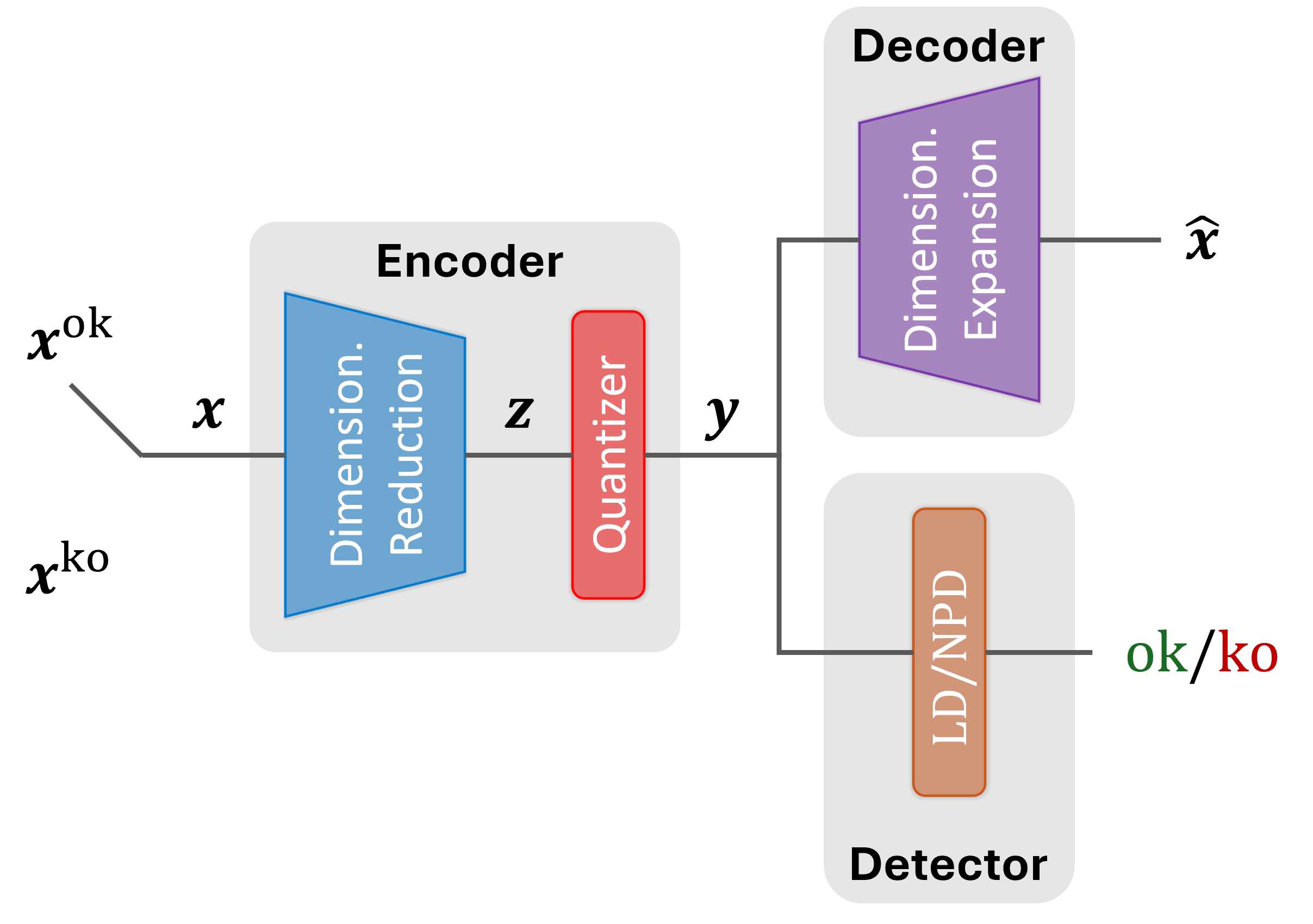} 
    \caption{Block diagram of the compression based on Random Component Selection. }
    \label{fig:RCS-block-diagram}
\end{figure}

\begin{figure*}[t!]
\centering
  \subfloat[Rate-distinguishability and rate-distortion trade-offs for the LD.]{\includegraphics[]{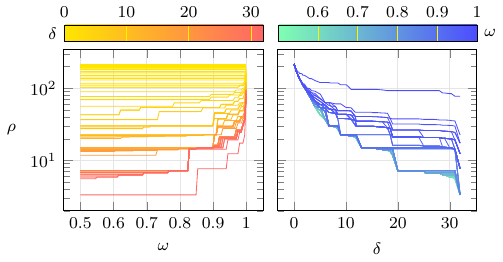} \label{RCS-tradeoffs-unsupervised-white}}
  \subfloat[Rate-distinguishability and rate-distortion trade-offs for the NPD.]{\includegraphics[]{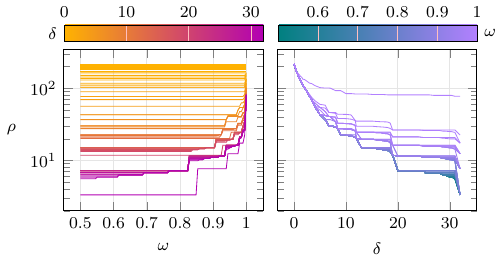} \label{RCS-tradeoffs-supervised-white}}
\caption{Rate-distortion-distinguishability trade-off administered by RCS compressor in case of white anomaly.}
\label{fig:RCS-white}
\end{figure*}

\subsection{Random Component Selection}

The RDD optimization problem is based on assumptions on signals and encoder. Here, we relax the assumption on the encoder by employing a KLT-based dimensional reduction followed by a quantization stage as depicted in Fig.~\ref{fig:RCS-block-diagram}. Here, the dimensionality reduction contributes to the distortion by projecting the input $\x$ onto the subspace spanned by $m<n$ selected components, i.e., eigenvectors of $\covok$. On the other hand, the quantization stage has the purpose of limiting the rate. When the principal components are selected, this scheme is equivalent to PCA-based compression \cite{Burrello_JIoT2020, Marchioni_JIOT2023}, however, here we also explore the possibility of picking components that are different from the principal.
Since dimensionality reduction is a linear operation and the adopted quantization can be modeled as a sum of Gaussian variables, the compressed signals are still Gaussian.

The dimensionality reduction that minimizes distortion is based on PCA and selects the $m$ components corresponding to the largest eigenvalues of $\covok$. The introduction of distinguishability constraint prevents an analytic solution. Hence, we address the joint optimization of distortion and distinguishability with a heuristic approach that randomly selects $m$ components.
Specifically, we model the selection with a vector $\m$ that contains the indices of the $m$ randomly selected components of $\x$ so that the resulting signal $\z = (x_j | j\in \m)$.
Given that this procedure is linear and $\x$ is Gaussian, $\z$ is Gaussian. For instance, depending on the source of the input vector ($\x=\xok$ or $\x=\xko$), $\z$ is characterized by the following covariance matrices
\begin{equation}
    \covokxh = {\rm diag}\left(\left(\lambdaok_j| j \in \m\right)\right),  \quad\quad  \covkoxh = \left[ \covko \right]_\m
\end{equation}
where $\left[ \cdot \right]_\m$ indicates an operator that takes a square matrix and produces a sub-matrix with the rows and columns indexed by $\m$. Note that when $\covko$ is diagonal $\left[ \covko \right]_\m ={\rm diag}\left(\left(\lambdako_j| j \in \m\right)\right)$.

The subsequent quantization stage is only intended to limit the rate that would otherwise be infinite \cite{Bell_1995NC}. Therefore, we use a Gaussian source quantizer that is optimal in the sense of rate-distortion, introducing negligible distortion fixed at $\epsilon \ll \lambdaok_{n-1}$. In this sense, the quantizer behaves as Gaussian-additive, producing a quantized vector $\y = \z + \Deltav$, where $\Deltav$ is a zero-mean Gaussian vector with independent components having the same variance $\theta_\epsilon = \epsilon / m$ \cite[Theorem 13.3.3]{Cover_1991}.

Finally, the decoder reconstructs the signal $\xh$ from the compressed vector $\y$ as
\begin{equation}
    \hat{x}_j = \begin{cases}
        \label{eq:RCA-reconstructed}
        y_i \;\; \text{with} \;\; i|m_i=j & j \in \m \\
        0 & j \notin \m.
    \end{cases}
\end{equation}

In the considered scheme, the resulting distortion $D$ is the sum of the dimensionality reduction and quantization contributions. The former consists of the sum of the eigenvalues corresponding to the discarded components, while the latter is described by $\epsilon$ that is a degree of freedom that we set to be negligible such that
\begin{equation}
   D = \sum_{j \notin \m}\lambdaok_j + \epsilon. 
   \label{eq:distortion-RCC}
\end{equation}

On the other hand, the rate is ruled by the quantization stage. In fact, the rate $\rho=I(\xhok;\xok)$ corresponds to $I(\yok;\xok)$ because $\xhok$ and $\yok$ share the same information. Moreover, since $\yok$ is obtained from $\zok$ that, in turn, is obtained from $\xok$, the information shared between $\yok$ and $\zok$ is the same as that shared between $\yok$ and $\xok$, hence $I(\yok;\xok) =
I_\epsilon(\yok;\zok)$. Therefore, exploiting \cite[Theorem 13.3.3]{Cover_1991}, the rate can be expressed as
\begin{equation}
\rho(\epsilon) = I_\epsilon(\xhok;\xok) = I_\epsilon(\yok;\zok) =
\frac{1}{2}\sum_{j \in \m}\log_2\frac{\lambdaok_j}{\theta_\epsilon}
\label{eq:rate-RCC}
\end{equation}
where we highlight the dependence on $\epsilon$, which is a degree of freedom\footnote{
We want to point out that the choice of $\epsilon$ is irrelevant as long as $\theta_\epsilon = \epsilon /m < \lambdaok_{n-1}$, which means that none of the $\zok$ components are completely distorted. To prove it we rely on an alternative expression of \eqref{eq:rate-RCC} \cite{Kolmogorov_TIT1956}
\begin{equation*}    
I_\epsilon(\yok;\zok) = \frac{m}{2} \log \frac{m}{2 \epsilon^2\pi e } + h(\zok) + o(1).
\end{equation*}
It is evident how the rate depends mainly on the dimensionality of the subspace $m$, while the differential entropy $h(\zok)$ accounts for the variability of the rate due to different combinations of $m$ components. The value of $\epsilon$ only determines an offset. In fact, $I_{\epsilon_2}(\xhok;\xok) = I_{\epsilon_1}(\xhok;\xok) - m \left( \log \nicefrac{1}{\epsilon_1} - \log \nicefrac{1}{\epsilon_2} \right)$.}.

We carried out numerical evidence in which we fixed $\epsilon = \nicefrac{\lambdaok_{n-1}}{10^2}$ and randomly selected $M=\min\left\{ {{n}\choose{m}}, 10^4\right\}$ vectors $\m$ for each $m \in \{1, \dots, n-1\}$. For each encoder, we measure the rate and distortion with \eqref{eq:rate-RCC}, and \eqref{eq:distortion-RCC}, respectively. As in the previous example, we employ  $\covko=\mat{I}_n$ as an anomaly and $P_{\rm det}$ to measure the performance of LD and NPD. 
Fig.~\ref{fig:RCS-white} illustrates the obtained trade-offs. The profiles confirm the trends observed with the optimal Gaussian additive encoder in Fig.~\ref{fig:RDD-white-PD} for which the cost of improving supervised and unsupervised detection performance is either a higher rate or a higher distortion. However, this compression scheme is suboptimal and the cost in rate necessary to achieve the same distortion-distinguishability performance is inevitably higher than the Gaussian-additive optimized with \eqref{eq:q}.

\subsection{JPEG Compression and detection}

\begin{figure}
\centering
  \includegraphics[scale=0.3]{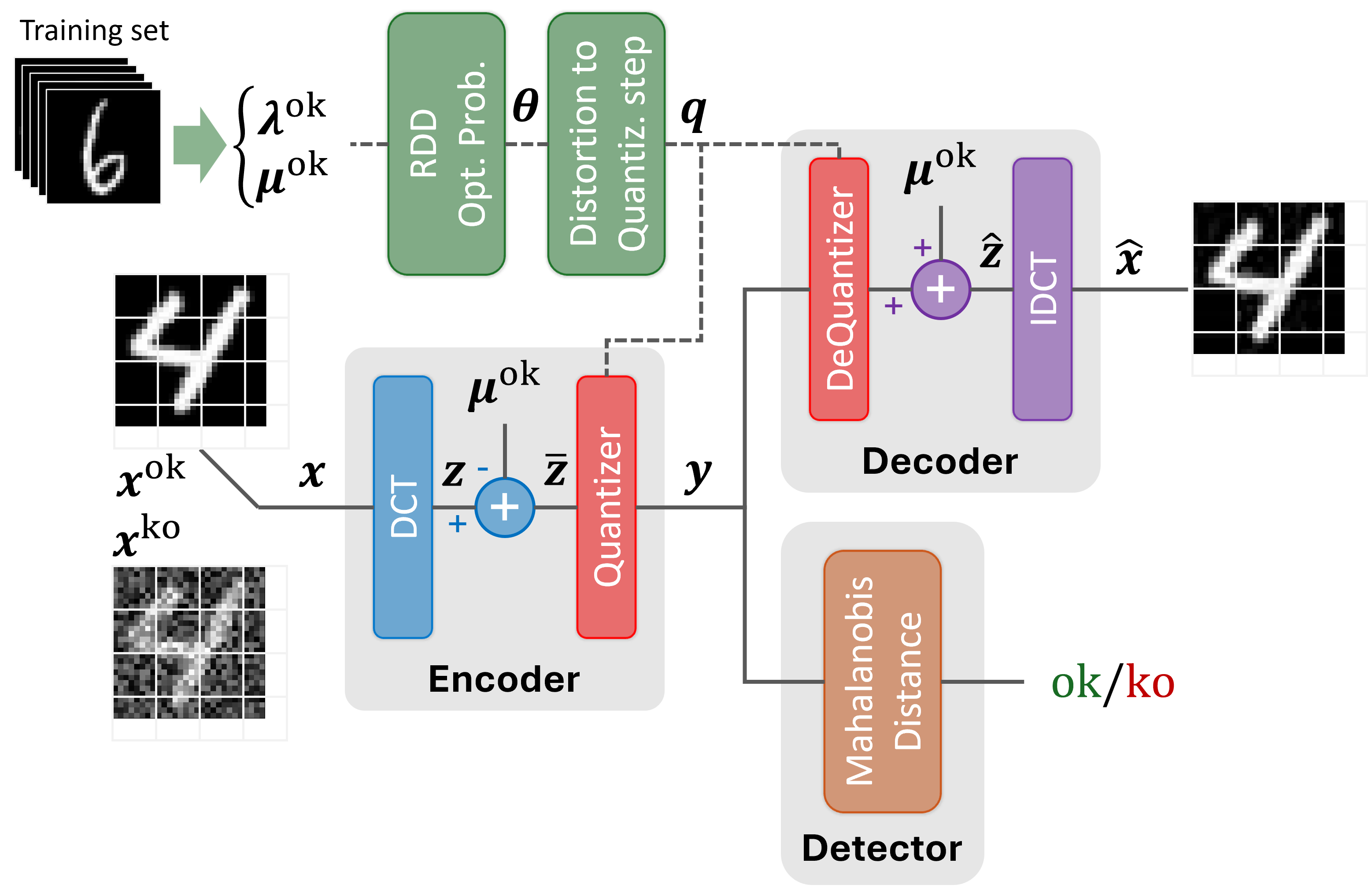}
\caption{JPEG block diagram.}
\label{fig:jpeg}
\end{figure}

While in the previous section, we relaxed the assumption on the encoder, here we also relax the one on signals. We apply the RDD framework to a standard compression algorithm to enhance the detection performance. We choose the JPEG \cite{Wallace_TCE1992} which is a standard for lossy compression of both color and grayscale images. For the sake of simplicity, we focus on the grayscale images of the MNIST dataset \cite{Deng_2012SPM}.

JPEG compression of grayscale images comprises a few stages. Images are first split into square blocks of $8\times 8$ pixels, and each block goes through Discrete Cosine Transformation (DCT) independently. The resulting DCT coefficients are quantized following a quantization table that depends on the desired quality of reconstruction. Finally, quantized coefficients of all blocks are encoded with an entropy encoder. Decoding is achieved through an entropy decoder that returns the quantized DCT coefficients to a dequantizer and a block performing the inverse DCT (IDCT) to obtain the blocks composing the reconstructed image.

Since the quantizer is the only stage responsible for the loss of information, it is also the one to work on to test the efficacy of the proposed framework. The idea consists in modifying the quantization table depending on the distortion and distinguishability constraints as in the
optimization problem in \eqref{eq:rate-distortion-distinguishability}. We then limit this example to the unsupervised scenario in which the compressor is unaware of the anomaly so that we only consider the first set of constraints \eqref{eq:rate-distortion-distinguishability}. 

We then assume that the DCT coefficients of each block are independent of each other and distributed as a Gaussian random variable with mean $\muok_j$ and variance $\lambdaok_j$ where $j=0,1,\dots,63$ is the index of the DCT coefficient. Let $\x$ be a block in a sample image of the training set and $\z={\rm DCT}(\x)$ the corresponding DCT coefficients. We estimate the mean and the variance of the $j$-th DCT coefficient as:
\begin{equation}
    \muok_j = \sum_{i=0}^{N-1} \sum_{l=0}^{B-1} z^{(i,l)}_j,
\;\;\;
    \lambdaok_j = \sum_{i=0}^{N-1} \sum_{l=0}^{B-1} \left[z^{(i,l)}_j - \muok_j\right]^2
\end{equation}
where $z^{(i,l)}_j$ is the $j$-th DCT coefficient of the $l$-th block in the $i$-th image, $N$ is the number of images in the training set, and $B$ is the number of blocks in an image.

By removing the mean $\muvok$ from $\z$, we obtain a zero-mean Gaussian signal $\zvb = \z - \muvok$ that allows us to turn the optimization problem \eqref{eq:rate-distortion-distinguishability} into \eqref{eq:q}-\eqref{eq:c2-a}. The solution consists of a set of parameters $\vec\xi = (\xi_0, \dots, \xi_{63})$ that minimizes the rate $\rho$ given a constrain in distortion $\delta$ and distinguishability $\omega$. From $\xi_j$, we obtain the average distortion $\theta_j = \lambdaok_j(1-\xi_j)$ that quantization induces to $\zb_j$ to minimize the rate and meet the distortion and distinguishability constraints. 

The mapping between the distortion and quantization steps is empirically estimated from the training set. This mapping allows us to map $\vec\theta=(\theta_0, \dots, \theta_{63})$ to a quantization table $\q$, where each $\theta_j$ corresponds to the quantization step $q_j$ to apply to $\zb_j$ and obtain its quantized version $y_j = \lfloor \zb_j/q_j \rceil$, where $\lfloor \cdot \rceil$ is the operator that rounds the argument to the nearest integer. Fig.~\ref{fig:jpeg} shows a block diagram of a JPEG encoder and decoder in which the quantization table $\q$ is obtained from the RDD optimization problem \eqref{eq:q}-\eqref{eq:c2-a}\footnote{Note that we have omitted the entropy encoder and decoder from the scheme, as entropy coding is not implemented due to its lossless nature.}. 

In this scheme, we measure distortion as the mean square error between $\xok$ and $\xhok$ and the rate as the Shannon entropy of $\yok$ (as in \cite{Agustsson_Neurips2017}):
\begin{equation}
\rho = I(\xok;\xhok) \simeq H(\yok) \simeq - n {\sum}_k\hat{p}_k \log{\hat{p}_k}
\label{eq:shannon-entropy}
\end{equation}
where $\hat{p}_k$ being the estimate of the probability for the entries of $\y$ to be equal to the $k$.

\captionsetup[subfigure]{width=0.3\columnwidth}
\begin{figure}
\centering
    \subfloat[\centering $\vec\lambda^\ok$]{
        \includegraphics{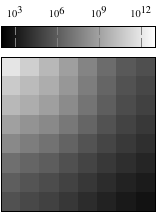}}
    \subfloat[\centering $\q$ with $\omega=0$]{
        \includegraphics{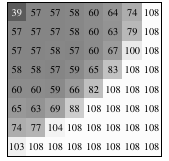}}
    \subfloat[\centering $\q$ with $\omega=10^3$]{
        \includegraphics{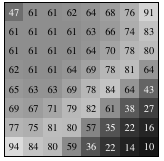}}
\caption{Variance of the DCT coefficients (a) and quantization tables as a result of two different values of the distinguishability constraint and same distortion $\delta=0.3$ (b) and (c). }
\label{fig:qtable}
\end{figure}

Considering the MNIST dataset with $N=60\,000$ training samples of $28\times 28$ images, each image is split into $B=16$ blocks. Fig.~\ref{fig:qtable}-(a) shows the variance of the DCT coefficients $\lambdavok$ estimated from the training set, while Figs.~\ref{fig:qtable}-(b) and (c) show two examples of the quantization table obtained with $\delta=0.3$ and $\omega=0$ (no constraint on distinguishability) and $\omega=10^3$, respectively. MNIST dataset, similarly to other types of images, manifests low-pass characteristics with an average high magnitude of the DCT coefficient related to low frequencies (upper left corner in Fig.~\ref{fig:qtable}-(a)) and a lower average magnitude of the high-pass coefficients (lower right corner). Following what was observed in Fig.~\ref{fig:RDD-profiles}, traditional rate-distortion applies the same distortion to all components, resulting in low frequencies that are relatively less distorted than high frequencies. We also observe this behavior in Fig.~\ref{fig:qtable}-(b) consistent with standard JPEG quantization tables \cite{Wallace_TCE1992}. In contrast and coherently with Fig.~\ref{fig:RDD-profiles}, Fig.~\ref{fig:qtable}-(c) shows that the distinguishability constraint leads to high frequencies being less distorted than medium frequencies, even though the latter is on average more informative. Both quantization tables in Fig.~\ref{fig:qtable}-(b) and Fig.~\ref{fig:qtable}-(c) adapt JPEG compression to the MNIST dataset. However, the latter should allow a detector to detect anomalies that affect the input given the same rate and distortion. 

For the assessment of the different tables, we consider an anomaly source obtained as a uniform noise mixed with the MNIST dataset and employ the Mahalanobis distance \cite{Maesschalck_20001CILS} as a detector. Since JPEG is a block-wise compression, the detector processes the DCT coefficients corresponding to blocks of compressed images and discriminates whether each block corresponds to a corrupted image. The Mahalanobis distance requires the mean and covariance of the DCT coefficients of a compressed block, which are estimated from the compressed images of the training set.

\begin{figure}
\centering
\includegraphics[]{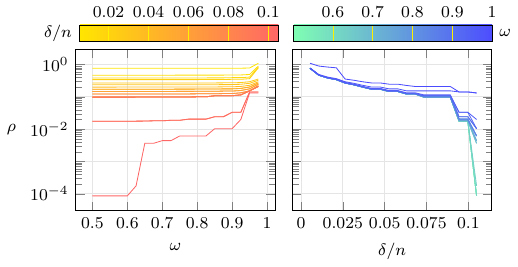}
\caption{Rate-distortion-distinguishability trade-off with JPEG adapted to MNIST dataset and uniform noise mixed with the signal as an anomaly.}
\label{fig:parto_jpeg}
\end{figure}

Fig.~\ref{fig:parto_jpeg} shows the Pareto curve for the JPEG compression and the Mahalanobis-based detector when uniform noise is mixed with the input image in which the rate is estimated as the Shannon entropy of the quantized DCT coefficients. The trends agree with those previously observed, in which improvement in detection or lower distortion comes at the cost of a higher rate. This result shows the relevance of the rate-distortion-distinguishability trade-off even in the case of images as signal and JPEG as a compressor.

\subsection{End-to-end image compression and detection}

\begin{figure}
\centering
\includegraphics[scale=0.3]{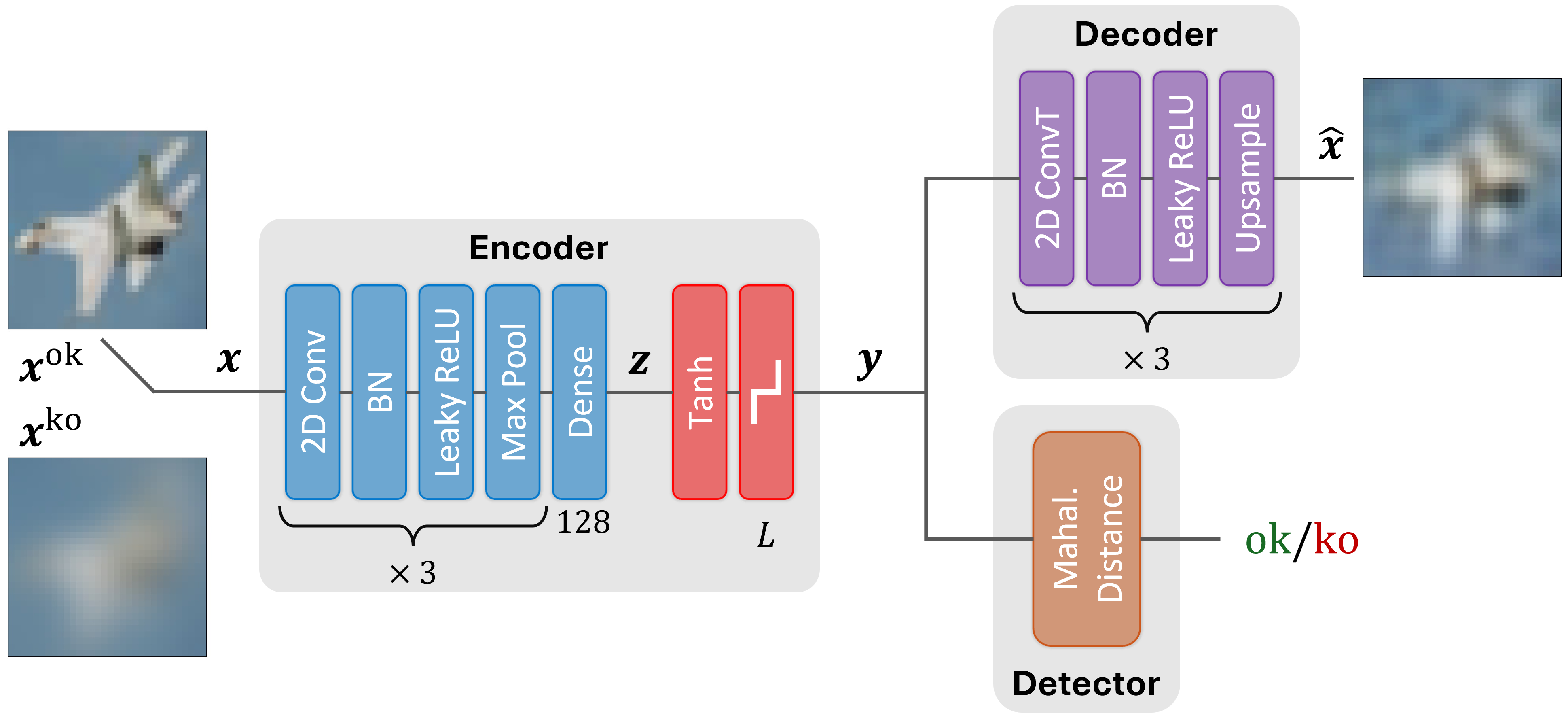}
\caption{System for Neural Network-based CIFAR-10 image compression and detection.}
\label{fig:neural_compressor}
\end{figure}

Compression mechanisms based on AEs have recently gained popularity for the possibility of end-to-end optimization of the encoder and decoder. On this line, we explore the simultaneous training of an encoder, a decoder, and a detector. 

An AE alone performs dimensionality reduction. Hence, for proper compression, we need to add a quantization layer in the bottleneck that quantizes the latent vector and limits the rate. We follow the quantization procedure proposed in \cite{Agustsson_Neurips2017} that has already been adopted in other AE-based compression mechanisms \cite{Mentzer_CVF2018, Blau_ICML2019}. Quantization consists of squeezing the values of the latent representation in the range $[-1, 1]$ with a $\tanh$ and performing a uniform quantization of each component with $L$ levels.

In this example, we focus on the CIFAR-10 dataset \cite{Krizhevsky_2009} and use the AE architecture proposed in \cite{Ruff_PMLR2018} to which we add the quantization layer. Specifically, the encoder reduces the $32\times 32 \times 3$ input image ($n=3072$) to an $m=128$ dimensional vector $\z$ and then apply quantization so that the compressed signal $\y$ is computed as
\begin{equation}
\label{eq:quatization_forward}
y_j = {\rm arg} \min_k \left| \tanh \left( z_j \right) - l_k\right|
\end{equation}
where $l_k=-1+2k/(L-1)$ is the $k$-th quantization level.

Since the expression in \eqref{eq:quatization_forward} is not differentiable, during the backward pass it is approximated with \cite{Mentzer_CVF2018}
\begin{equation}
\label{eq:quatization_back}
    y_j = \sum_{k=0}^{L-1}\frac{e^{-\left|\tanh \left( z_j \right) - l_k\right|}}{\sum_{t=0}^{L-1} e^{-\left|\tanh \left( z_t \right) - l_t\right|}} l_k.
\end{equation}

The resulting vector $\y$ is fed to the decoder, which outputs the reconstructed input $\hat{\x}$, and the detector, which computes the anomaly score as the Mahalanobis distance of $\y$.  The detection capabilities of the detector are promoted by a regularization term acting on $\| \z \|^2$. This approach is inherited from the Deep Support Vector Data Description \cite{Ruff_PMLR2018} and Shrink AE \cite{Cao_2019TC} anomaly detectors in which the encoder performs dimensionality reduction optimized for detection by minimizing the volume of the latent representation of the normal signal $\zok$. The complete system is depicted in Fig.~\ref{fig:neural_compressor} and it is trained with the following loss function:
\begin{align}
\label{eq:loss}
\mathcal{L}_{\beta}(\xok) &= {\rm MSE}(\xok, \xhok) + \frac{\beta}{m} \| \zok \|^2
\end{align}
where $\beta$ is the regularization weight. With such a loss function, the encoder can be optimized for both reconstruction (by minimizing {\rm MSE}) and anomaly detection (by minimizing the regularization term).

While distortion is measured as the mean squared error between $\xok$ and $\xhok$, the rate is approximated with the Shannon entropy of $\yok$ as in \eqref{eq:shannon-entropy}
where the dimension of $\yok$ is now $m$ and $\hat{p}_k$ is the estimate of the probability for the entries of $\yok$ to be equal to the $k$-th quantization level $l_k$.

As anomalies, we use CIFAR-10 corruptions defined in \cite{Hendrycks_2019ICLR} and employed for anomaly detection performance benchmark in \cite{Ruff_2021ProcIEEE}. In particular, we select the Gaussian blur as the target anomaly. 

We train\footnote{Each training is performed using the Adam optimizer with the weight of $l_2$ weight regularization set to $10^{-6}$, batch size of $200$, and an initial learning rate of $10^{-4}$ that is lowered whenever the loss reaches a plateau with patience of $20$ epochs.} several AEs characterized by a different combination of $L \in \{2, 2^2, \dots, 2^9\}$, which acts on the rate, and $\beta \in \{0, 10^{-6}, 10^{-5} \dots, 10^3\}$, which administers the trade-off between distortion and detection.
The resulting three-fold trade-off is illustrated in Fig.~\ref{fig:AE-trade-off}. Although no assumption is made on the signals and the compressor, the distinguishability dependency on rate and distortion, confirms the theoretical trends.

\begin{figure}
\centering
  \includegraphics[]{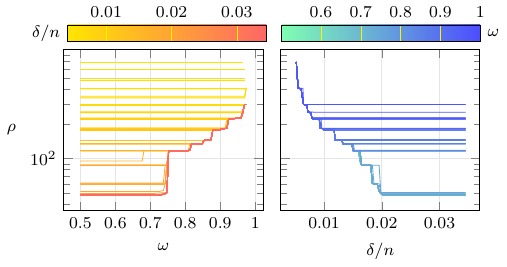} \label{fig:AE-unsupervised}
\caption{Rate-distortion-distinguishabilty trade-off for Neural Network-based system working with CIFAR-10 images in the Gaussian blur anomaly-agnostic case.}
\label{fig:AE-trade-off}
\end{figure}

\section{Conclusion}
\label{sec:conclusion}

In applications where anomalies have to be detected from compressed signals, a trade-off between rate, distortion, and distinguishability naturally emerges.
We addressed such trade-off with a framework that extends the classical rate-distortion theory by adding a distinguishability constraint.

Under the assumption of Gaussian signals and considering a Gaussian-additive encoder, we derived the optimization problem that allows us to observe and discuss the Pareto surface of the three involved quantities. Previous works have shown that optimizing solely for distortion given a rate budget negatively affects the distinguishability between normal and anomalous compressed signals. In this work, we have shown that to preserve distinguishability and improve detection performance, one must increase either the rate or the distortion. For instance, given a distortion constraint, a compressor preserves distinguishability by allocating some rate to signal components that provide the least information for reconstruction but are relevant for detection.

We assess how the proposed framework generalizes to different scenarios with three examples that involve: a compression scheme based on linear dimensionality reduction followed by quantization, a variation of JPEG compressor to include distinguishability, and a compression mechanism based on an autencoder optimized for both reconstruction and detection. All these three examples confirmed the trends observed in the theoretical setting.

\bibliographystyle{IEEEtran}
\bibliography{reference}

\end{document}